\documentclass[a4paper,twocolumn,11pt,accepted=2024-01-08]{quantumarticle}
\pdfoutput=1
\usepackage[utf8]{inputenc}
\usepackage{amsopn}
\usepackage{amsmath}
\usepackage{amssymb}
\usepackage{amsthm}
\usepackage{hyperref}
\usepackage{graphicx}
\usepackage{color}
\usepackage{listings}
\usepackage{subcaption}
\usepackage{enumerate}
\usepackage{cite}
\usepackage{mathtools}
\usepackage[numbers]{natbib}

\usepackage{multirow}
\usepackage{adjustbox}
\usepackage{geometry}
\usepackage{qcircuit}

\usepackage{tikz}
\usetikzlibrary{angles,quotes}
\usetikzlibrary{decorations.pathreplacing}

\DeclareMathOperator{\tr}{tr}
\DeclareMathOperator{\Tr}{Tr}

\newcommand{\ket}[1]{\ensuremath{|#1\rangle}}
\newcommand{\bra}[1]{\ensuremath{\langle#1|}}
\newcommand{\ketbra}[2]{\ensuremath{\ket{#1} \! \bra{#2}}}
\newcommand{\proj}[1]{\ensuremath{\ketbra{#1}{#1}}}
\newcommand{\braket}[2]{\ensuremath{\langle{#1}|{#2}\rangle}}

\renewcommand{\1}{{\rm 1\hspace{-0.9mm}l}}

\newcommand{\XX}{\mathcal{X}}

\newcommand{\YY}{\mathcal{Y}}
\newcommand{\ZZ}{\mathcal{Z}}

\newcommand{\DD}{\mathcal{D}}
\newcommand{\TT}{\mathcal{T}}
\newcommand{\CC}{\mathcal{C}}
\newcommand{\PP}{\mathcal{P}}

\newcommand{\UU}{\mathcal{U}}
\newcommand{\LL}{\mathcal{L}}

\newcommand{\ketV}[1]{\ensuremath{|#1\rangle\!\rangle}}
\newcommand{\braV}[1]{\ensuremath{\langle\!\langle#1|}}
\newcommand{\ketbraV}[2]{\ensuremath{\ketV{#1}\braV{#2}}}
\newcommand{\projV}[1]{\ensuremath{\ketbraV{#1}{#1}}}

\theoremstyle{definition}

\newtheorem{lemma}{Lemma}
\newtheorem{theorem}{Theorem}

\newtheorem{remark}{Remark}
\newtheorem*{rem*}{Remark}

\newcommand{\rhozeropsi}{\ensuremath{\rho_0^{(\psi)}}}
\newcommand{\rhoonepsi}{\ensuremath{\rho_1^{(\psi)}}}
\newcommand{\pcondone}{p_{\mathrm{I}}^{(\psi, \Omega)}}
\newcommand{\pcondtwo}{p_\mathrm{II}^{(\psi, \Omega)}}
\newcommand{\pone}{p_\mathrm{I}}
\newcommand{\ptwo}{p_\mathrm{II}}

\usepackage{caption}
\usepackage{multirow}
\usepackage{hhline}

\def\>{\rangle}
\def\<{\langle}

\begin{document}

\title{Discrimination and certification of unknown quantum measurements}

\author{Aleksandra Krawiec}
\email{aaleksandra.krawiec@gmail.com}
\affiliation{Institute of Theoretical and Applied Informatics, Polish Academy of Sciences, ul. Ba{\l}tycka 5, 44-100 Gliwice, Poland}
\affiliation{AstroCeNT, Nicolaus Copernicus Astronomical Center, Polish Academy of Sciences, ul. Rektorska 4, 00-614 Warsaw, Poland}
\orcid{0000-0001-8390-6569}
\author{{\L}ukasz Pawela}
\affiliation{Institute of Theoretical and Applied Informatics, Polish Academy of Sciences, ul. Ba{\l}tycka 5, 44-100 Gliwice, Poland}
\orcid{0000-0002-0476-7132}
\author{Zbigniew Pucha{\l}a}
\affiliation{Institute of Theoretical and Applied Informatics, Polish Academy of Sciences, ul. Ba{\l}tycka 5, 44-100 Gliwice, Poland}
\orcid{0000-0002-4739-0400}

\maketitle

\begin{abstract}
We study the discrimination of von Neumann measurements in the scenario when we
are given a reference measurement and some other measurement. The aim of the
discrimination is to determine whether the other measurement is the same as the
first one. We consider the cases when the reference measurement is given without
the classical description and when its classical description is known. Both
cases are studied in the symmetric and asymmetric discrimination setups.
Moreover, we provide optimal certification schemes enabling us to certify a
known quantum measurement against the unknown one.
\end{abstract}

\maketitle

\section{Introduction}
The need for appropriate certification tools is one of the barriers to the
development of large-scale quantum technologies.~\cite{eisert2020quantum} In
this work, we propose tests that verify if a given device corresponds to either
its classical description or the reference device.

But why should we care about the discrimination of devices which descriptions we
do not know? A lot is known about discrimination of quantum states, channels and
measurements, which descriptions we do know. In the standard discrimination
problem, there are two quantum objects, and one of them is secretly chosen. The
goal of discrimination is to decide which of the objects was chosen.  
(See~\cite[Chapter 11]{paris2004quantum}, \cite{bergou2007quantum,barnett2009quantum,
bae2015quantum} for the pedagogical reviews  of the discrimination problems and
their applications). These objects can be quantum states but also quantum
channels and measurements. However, what if we were given a reference quantum
measurement or channel instead of its classical description? Then, we may want
to discriminate them regardless of their classical descriptions. Therefore, we
arrive at the new problem of discrimination of unknown objects.

Discrimination of known quantum channels was mainly studied for certain classes
of channels like unitary
channels~\cite{acin2001statistical,bae2015discrimination,kawachi2019quantum}.
Advantage of using entangled states for minimum-error discrimination of quantum
channels was studied
in~\cite{sacchi2005optimal,sacchi2005entanglement,piani2009all}. General
conditions when quantum channels can be discriminated in the minimum error,
unambiguous and asymmetric scenarios, were derived in~\cite{duan2009perfect},
\cite{wang2006unambiguous} and~\cite{krawiec2021excluding} respectively. Another
formalism used for studying discrimination of quantum channels is based on
process POVM (PPOVM)~\cite{ziman2008process}. It was applied to discrimination
of unitary channels in~\cite{sedlak2009unambiguous,ziman2010single}. A practical
approach to discrimination of noisy quantum gates without assistance of
entanglement was recently studied in~\cite{choi2022single}.

Discrimination of unknown unitary channels was first studied in the
work~\cite{hillery2010decision} in both minimum-error and unambiguous setups.
The authors calculated that the probability of successful minimum-error
discrimination between two random qubit unitary channels equals $7/8$ and they
made use of the input state $\ket{\psi_-} =\frac{1}{\sqrt{2}} \left(
\ket{01}-\ket{10} \right)$.  The authors of~\cite{soeda2021optimal}  proved that
the probability $7/8$ is optimal in the sense that it cannot be improved by the
use of any (even adaptive) discrimination strategy for the qubit case. Recent
results concerning discrimination of unknown unitary channels can be found
in~\cite{hashimoto2022comparison}.

In this work we are mainly interested in the discrimination of von Neumann
measurements, which are the most commonly used in practical applications. Von
Neumann measurements are a special subclass of general projective measurements,
in the sense that they are fine-grained measurements. These measuements are at
the crux of most schemes and protocols appearing in quantum information and
computing. Furthermore, thanks to the Naimark
construction~\cite{watrous2018theory} any measurement can be implemented as a
projective measurement on a larger Hilbert space. Minimum error discrimination
of von Neumann measurements was studied in
single-shot~\cite{puchala2018strategies} and
multiple-shot~\cite{puchala2021multiple} regimes. Asymmetric discrimination of
measurements was studied in~\cite{lewandowska2021optimal}. The advantage of
using entangled stated for single-shot discrimination between qubit measurements
was experimentally shown in~\cite{mikova2014optimal}. Application of process
POVMs for discrimination of quantum measurements can be found
in~\cite{ziman2009unambiguous,sedlak2014optimal}.

Possible extensions of the study of discrimination of quantum states and
measurements include discrimination between process matrices which are a
universal method defining a causal structure. For
instance~\cite{lewandowska2023strategies} provides an exact expression for the
optimal probability of correct distinction of process matrices. Finally, the
problem of discrimination of quantum states has also been studied in much
broader approach in the case of general contextual and non-contextual
theories~\cite{flatt2022contextual}.

In this work we study discrimination of unknown von Neumann measurements in
symmetric and asymmetric scenarios. We begin with preliminaries in
Section~\ref{sec:preliminaries} and detailed setups for symmetric and asymmetric
discrimination of quantum measurements will be presented therein. Next, we will
study the problem when one of the measurements is given without classical
description and we want to verify if the other measurement is a copy of the same
measurements or it is some other one. This problem will be studied in
Section~\ref{sec:discrimination_both_unknown_measurements}. Later, we will
assume that one copy of a measurement is given with its classical description
and we want to know whether the other measurement is a copy of the same
measurement. This problem will be studied in
Section~\ref{sec:one_fixed_measurement}. We will conclude in
Section~\ref{sec:conclusion}.

\section{Preliminaries}\label{sec:preliminaries}
Let $\XX$, $\YY$ and $\ZZ$ be Hilbert spaces where $\dim(\XX) = \dim(\YY) = d$,
$\dim(\ZZ) = d^2$. Let $\LL(\XX)$ be a set of linear operators acting from $\XX$
to $\XX$. Let $\UU(\XX)$ denote the set of unitary operators. Let $\DD(\XX)$
denote the set of quantum states, $\CC(\XX)$ denote the set of quantum channels
and $\TT(\XX)$ denote the set of quantum operations. For $U \in \UU(\XX)$, a
unitary channel will be denoted $\Phi_U(\cdot) \coloneqq U \cdot U^\dagger$. We
will also utilize two special quantum channels. The first one is the
depolarizing channel, which transforms every quantum state into the maximally 
mixed state. Formally, it is defined  for $X \in \LL(\XX)$ as
\begin{equation*}\label{eq:depolarizing_channel_df}
\Phi_*(X) \coloneqq \Tr(X) \frac{\1}{\dim(\XX)}.
\end{equation*}
The second one is the dephasing channel defined as
\begin{equation*}\label{eq:dephasing_channel_df}
\Delta(X) \coloneqq \sum_i \proj{i} X \proj{i}.
\end{equation*}

A quantum measurement is defined as a collection of positive semidefinite 
operators $\PP = \{ E_1, \ldots ,E_m\}$ which satisfy $\sum_{i=1}^m = \1$, where 
$\1$ is the identity operator. Operators $E_i$ are called \emph{effects}. 
When a quantum state $\rho$ is measured by the measurement $\PP$, then we  
obtain a label $i$ with probability $p(i) = \tr \left(E_i \rho\right)$ and the 
state $\rho$ ceases to exist. 
We will be 
particularly interested in von Neumann measurements, which effects are of the 
form $\PP_U = \{ \proj{u_1}, \ldots ,\proj{u_d}\}$, where $\ket{u_i} = U 
\ket{i}$ is the $i$-th column of the unitary matrix $U$. 
Every quantum measurement can be associated with a quantum channel 
\begin{equation}
\PP (\rho)= \sum_i \proj{i} \tr (E_i \rho),
\end{equation}
which outputs a diagonal matrix  where $i$-th entry on the diagonal corresponds to the probability of obtaining $i$-th label.

The Choi-Jamiołkowski representation of quantum operation $\Psi \in 
\mathcal{T}(\XX)$ is defined as $J \left(\Psi\right) \coloneqq \left(\Psi 
\otimes \1_\XX \right)(\projV{\1})$, where $\1_\XX$ is the identity channel on 
the space $\LL(\XX)$ and $\ketV{X}$ denotes the (lexicographical) vectorization 
of the operator $X$.

The diamond norm of a quantum operation $\Psi \in \TT(\XX)$ is defined as
\begin{equation}
\| \Psi \|_\diamond \coloneqq \max_{X: \|X\|_1=1} \left\| \left(\Psi \otimes 
\1_\XX\right) (X) \right\|_1,
\end{equation}
where $\1_\XX$ is, as previously, the identity channel on the space $\LL(\XX)$.
We will often use the bounds on the diamond 
norm~\cite{nechita2018almost,watrous2018theory}
\begin{equation}\label{eq:diamond_bounds}
\frac{1}{d} \|  J(\Psi) \|_1 \leq \| \Psi \|_\diamond 
\leq \|  \Tr_1 |J(\Psi)| \|.
\end{equation}

In this work we will focus on two approaches to discrimination of quantum 
measurements, which are symmetric and asymmetric discrimination.

\subsection{Symmetric discrimination}\label{sec:symmetric_scheme}
The goal of symmetric discrimination is to maximize the probability of correct
discrimination. It is also known as minimum-error discrimination. The schematic 
representation of symmetric discrimination of quantum measurements is depicted 
in Figure~\ref{fig:discrimination_of_measurements}.

\begin{figure}[h!]
\centering
\includegraphics[scale=1.4]{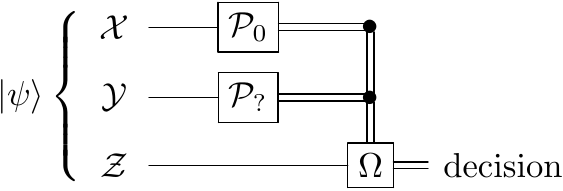}
\caption{Entanglement-assisted discrimination of von Neumann measurements}
\label{fig:discrimination_of_measurements}
\end{figure}

There are two black boxes. In the first black box there is a measurement
$\PP_0$. In the second box there is a measurement $\PP_?$, which can either be the 
same measurement $\PP_0$, or some other measurement, $\PP_1$. In other words 
$\PP_? \in \{ \PP_0, \PP_1 \}$. As the input state to the discrimination 
procedure we will take
a state $\ket{\psi} \in \XX \otimes \YY \otimes \ZZ$ and we will write $\psi
\coloneqq \proj{\psi}$ for the sake of simplicity. The measurement in the first
black box acts on the register $\XX$ and the second black box acts on the
register $\YY$. Basing on the outcomes of both measurements in the black boxes,
we prepare a final measurement on the register $\ZZ$. Having the output of the
final register, we make a decision whether $\PP_? = \PP_0$ or $\PP_? = \PP_1$.

To calculate the probability of the successful discrimination between quantum
measurements, we will make use of the Holevo-Helstrom 
theorem.~\cite{helstrom1969quantum}
 It states that the optimal probability of successful discrimination between any quantum 
channels $\Psi_0$ and $\Psi_1 \in \CC(\XX)$ is upper-bounded by
\begin{equation}\label{eq:HH}
p_{succ} \leq \frac{1}{2} + \frac{1}{4} \left\|  \Psi_0 - \Psi_1 
\right\|_\diamond
\end{equation}
and this bound can be saturated. This optimal probability of successful 
discrimination will be denoted $p^H_{succ} \coloneqq\frac{1}{2} + \frac{1}{4} 
\left\|  \Psi_0 - \Psi_1 \right\|_\diamond$.

\subsection{Asymmetric discrimination}\label{sec:asymmetric_scheme}
Asymmetric discrimination is based on the hypothesis testing. The null hypothesis
$H_0$ corresponds to the situation when $\PP_? = \PP_0$. The converse situation,
$\PP_? = \PP_1$ corresponds to alternative hypothesis $H_1$. The scheme of
asymmetric discrimination is as follows. We begin with preparing an input state
$\ket{\psi} \in \LL(\XX \otimes \YY \otimes \ZZ)$ and apply $\PP_0$ and $\PP_?$
on registers $\XX$ and $\YY$ respectively. Therefore, in the case when $\PP_? =
\PP_0$, we obtain as the output $\left(\PP_0 \otimes \PP_0 \otimes \1
\right)(\psi)$ and if $\PP_? = \PP_1$, then the output state yields $\left(\PP_0
\otimes \PP_1 \otimes \1 \right)(\psi)$. Having the output states, we prepare a
binary measurement $\left\{\Omega, \1 - \Omega\right\}$, where the effect
$\Omega$ accepts the null hypothesis and the effect $\1 - \Omega$ accepts the
alternative hypothesis.

The type I error (false positive) happens when we reject the correct null 
hypothesis.  In other words, this error happens when in both black boxes there were the same 
measurements, but we made a erroneous decision and said that in the measurements the black 
boxes were different. 
When the input state $\psi$ and measurement $\Omega$ are fixed, 
then the probability of making the type I error is given by the expression
\begin{equation}
\begin{split}
\pcondone & \coloneqq \Tr \left(  (\1 - \Omega) \left(\PP_0 \otimes \PP_0 
\otimes \1 \right)(\psi) \right) \\
& \phantom{:}= 1- \Tr \left( \Omega \left(\PP_0 \otimes \PP_0 \otimes \1 \right)(\psi) \right).
\end{split}
\end{equation}
The optimized probability of the type I error yields
\begin{equation}
\pone \coloneqq \min_{\psi, \Omega} \pcondone
\end{equation}
The probability of making the type II error (also known as false negative) for 
fixed input state and measurement equals 
\begin{equation}
\pcondtwo = \Tr \left( \Omega \left(\PP_0 \otimes \PP_1 \otimes \1 
\right)(\psi) \right) 
\end{equation}
and corresponds to the situation when we accept the null hypothesis when the 
alternative one was correct.  In other words, this error corresponds to the situation when the 
measurements in the black boxes were different but we made a mistake and said that in both black 
boxes there were the same measurements.
The optimized probability of making the type II error yields
\begin{equation}
\ptwo \coloneqq  \min_{\psi, \Omega} \pcondtwo.
\end{equation}

For both symmetric and asymmetric schemes we will study two cases. First we 
will assume that both measurements are unknown. Later, we will assume that we 
know the description of the reference measurement and the other measurement is 
unknown. We will be also interested whether the additional register is 
necessary for optimal discrimination. The summary of results is presented in the following table.

\begin{table}[h!]
\centering
\begin{tabular}{|p{1cm}||c|c|c|c|c|}
\hline
 & $p^H_{succ}$  & $p^H_{err}$ & $\pone$ & $\ptwo$ & ancilla\\
\hhline{|=|=|=|=|=|=|}
both unknown & $\frac{1}{2} + \frac{1}{2d}$  & $\frac{1}{2} - 
\frac{1}{2d}$ & $0$ & $1 - \frac{1}{d}$ & no \\
\hline
one fixed & $1 - \frac{1}{2d}$  & $\frac{1}{2d}$ & $0$ & $\frac{1}{d}$ & yes \\
\hline
\end{tabular}
\caption{Summary of for symmetric and asymmetric discrimination of unknown von 
Neumann measurements.}
\end{table}

We would like to remark that in the literature the task of \emph{certification}  often  refers to 
asymptotic asymmetric discrimination where objects are repeated in the iid manner, see for instance
~\cite{salek2022usefulness,wilde2020amortized,zhou2021asymptotic,cooney2016strong}. However, in 
this work we  will not consider the asymptotic case and focus on the scenario depicted in 
Fig.~\ref{fig:discrimination_of_measurements}.

\section{Discrimination of both unknown von Neumann 
measurements}\label{sec:discrimination_both_unknown_measurements}
In this section we will study a situation when we are given a von Neumann 
measurement $\PP_0$ but no classical description of it. This measurement will 
be our reference. We also have another von Neumann measurement $\PP_1$, which 
can be the same as the reference one, but it does not have to. In this section 
we will study the problem how to verify whether the second measurement is the 
same as the first one or not. Similar problem of discrimination of both unknown 
unitary channels was recently studied in~\cite{hashimoto2022comparison}.

\subsection{Symmetric discrimination}
We will be calculating the success probability for the discrimination of von
Neumann measurements in the scenario depicted in
Fig.~\ref{fig:discrimination_of_measurements}. Therefore, we will be actually
discriminating between $\PP_0 \otimes \PP_0$ and  $\PP_0 \otimes \PP_1$ in the
entanglement-assisted scenario. Thus, in order to use Holevo-Helstrom theorem we
will need to calculate the value of the diamond norm. As we do not have
classical description of either $\PP_0$ or $\PP_1$, we will assume that both
measurement are Haar-random, that is we will be discriminating between $\int 
\PP_U \otimes \PP_U dU$ and  $\int\PP_U \otimes \PP_V dU dV$. The probability 
of successful discrimination is formulated as the following theorem.

\begin{theorem}\label{thm:both_random_measurements_symmetric}
Let $\PP_0$ be a reference von Neumann measurement of dimension $d$ given 
without classical description. Let $\PP_1$ be another von Neumann measurement 
of the same dimension, also given without classical description. The optimal 
probability of correct verification if $\PP_1$ is the same as the reference 
channel in the scheme described in Subsection~\ref{sec:symmetric_scheme} 
equals  
\begin{equation}
p^H_{succ} = \frac{1}{2} + \frac{1}{2d}.
\end{equation}
\end{theorem}

\begin{remark}
The above theorem is a direct application of the Holevo-Helstrom 
Theorem (see Eq.~\eqref{eq:HH}) for discrimination between channels $\int 
\PP_U \otimes \PP_U dU$ and  $\int\PP_U \otimes \PP_V dU dV$, that is
\begin{widetext}
\begin{equation}
\begin{split}
p^H_{succ} &= \frac{1}{2} + \frac{1}{4} 
\left\|  \int 
\PP_U \otimes \PP_U dU- \int\PP_U \otimes \PP_V dU dV \right\|_\diamond
= \frac{1}{2} + \frac{1}{2d}.
\end{split}
\end{equation}
\end{widetext}
\end{remark}

\begin{proof}

Let $U \in \mathcal{U}(\XX), \ V \in \mathcal{U}(\YY)$ be unitary operators and
$\dim(\XX) = \dim(\YY) = d$. The probability of successful discrimination is
given by the Holevo-Helstrom theorem. To calculate this probability (Eq.
\eqref{eq:HH}), we need to calculate the diamond norm distance between the
averaged channels
\begin{equation}
\left\| \int \PP_U \otimes \PP_U dU  - 
\int \PP_U \otimes \PP_V dU dV  \right\|_\diamond. 
\end{equation}
As the von Neumann measurement $\PP_U$ can be seen as $\Delta 
\Phi_{U^\dagger}$, where 
$\Delta$ is a dephasing channel defined in Eq.~\eqref{eq:dephasing_channel_df},
we will actually be discriminating between
\begin{equation}
\begin{split}
& \int (\Delta \otimes \Delta)(\Phi_{U^\dagger} \otimes \Phi_{U^\dagger}) dU 
\quad \mathrm{and} \quad \\
& \int (\Delta \otimes \Delta)(\Phi_{U^\dagger} \otimes \Phi_{V^\dagger}) dU dV. 
\end{split}
\end{equation}

Using~\cite{puchala2017symbolic,collins2006integration} we calculate the 
Choi-Jamiołkowski representations of averaged unitary channels
\begin{widetext}
\begin{equation}\label{eq:j1_and_j2_introduced}
\begin{split}
&J \left(   \int \Phi_U \otimes \Phi_U dU \right)
=  \frac{1}{d^2-1} \left( \1 \otimes \1 + S \otimes S
\right) - \frac{1}{d(d^2-1)} \left(  S \otimes \1 + \1 \otimes S  \right), \\
&J\left(  \int \Phi_U \otimes \Phi_V dU dV  \right)
=  \frac{1}{d^2} \1 \otimes \1,
\end{split}
\end{equation}
\end{widetext}
where, unless said otherwise, $S$ is the Swap matrix of two $d$-dimensional systems and 
identity matrices are of dimension $d^2$.

Using the above, we can calculate the Choi-Jamiołkowski representations of the 
averaged measurements, that is
\begin{widetext}
\begin{equation}\label{eq:J_0}
J \left(\int \PP_U \otimes \PP_U dU\right) = 
\frac{1}{d^2-1} \left( 
\1 \otimes \left(   \1 - \frac{1}{d} S \right)
+ T \otimes \left( S - \frac{1}{d} \1 \right)
 \right)
\end{equation}
\end{widetext}
where $T \coloneqq \Delta(S)$, and
\begin{equation}
J \left(\int \PP_U \otimes \PP_V dUdV\right) =  
\frac{1}{d^2} \1 \otimes \1.
\end{equation}

For later convenience, we introduce $J \in \LL\left( \XX'\otimes \YY' \otimes \XX  \otimes \YY \right)$ as 
a difference of Choi  matrices of both randomized measurements, where the registers $\XX\otimes 
\YY$ correspond to input spaces and $\XX'\otimes \YY'$ correspond to the output spaces
\begin{equation}\label{eq:def_J_for_both_unknown_measurements}
\begin{split}
J &\coloneqq 
J \left(\int \PP_U \otimes \PP_U dU\right)
- J \left(\int \PP_U \otimes \PP_V dUdV\right) \\
&= 
\frac{1}{d^2-1} \left( 
\1 \otimes \left( \frac{1}{d^2}  \1 - \frac{1}{d} S \right)
+ T \otimes \left( S - \frac{1}{d} \1 \right)
 \right).
\end{split}
\end{equation}

The remaining part of the proof goes as follows. We will first calculate the
upper bound on the diamond norm $\| \int \PP_U \otimes \PP_U dU  - \int \PP_U
\otimes \PP_V dU dV \|_\diamond \leq \left\| \Tr_{\XX' \otimes \YY'} |J|
\right\|$ from Eq.~\eqref{eq:diamond_bounds} and next find an input state that
saturates it.

To calculate the upper bound we first need to find $|J| = \sqrt{J^\dagger J}$.
From Lemma~\ref{lm:j2=wj2_for_measurements} in Appendix~\ref{app:lemma},  taking
$W \coloneqq(2T-\1) \otimes S$ it holds that $(WJ)^2 = J^2$, and this gives a
polar decomposition of $J$ and we have $\left\| \Tr_{\XX' \otimes \YY'} |J|
\right\| = \left\| \Tr_{\XX' \otimes \YY'} WJ \right\|$. Hence:
\begin{widetext}
\begin{equation}
\begin{split}
\Tr_{\XX' \otimes \YY'} (WJ) 
&= \frac{1}{d^2-1} \Tr_{\XX' \otimes \YY'} \left(  \frac{1}{d} \1 \otimes \1  - 
\frac{1}{d^2} \1 \otimes S + \frac{d-2}{d} T \otimes \1 - \frac{d-2}{d^2} T 
\otimes S \right) \\
&= \frac{1}{d^2-1} \left(  \frac{d^2}{d}  \1  - \frac{d^2}{d^2} S + 
\frac{d(d-2)}{d}  \1 - \frac{d(d-2)}{d^2}  S \right) \\
&= \frac{1}{d^2-1} \left(  (2d-2)  \1  - \frac{2d-2}{d} S  \right) 
= \frac{2}{d+1} \left( \1  - \frac{1}{d} S  \right) \\
\end{split}
\end{equation}
\end{widetext}
and eventually we have
\begin{equation}
\begin{split}
\left\|  \Tr_{\XX' \otimes \YY'} |J| \right\| 
& = \left\Vert   \frac{2}{d+1} \left( \1  - \frac{1}{d} S  \right) \right\Vert \\
& = \frac{2}{d+1} \left\Vert  \1  - \frac{1}{d} S  \right\Vert
= \frac{2}{d}.
\end{split}
\end{equation}

Now we proceed to proving that the upper bound is saturated.
Let us take an input state $\LL(\XX \otimes \YY) \ni \rho \coloneqq \proj{a} $ which satisfies
\begin{equation}
S \rho = -\rho, \quad   \rho^\top = \rho .
\end{equation}
As the vector $\ket{a}$ we can take any $\XX \otimes \YY \ni \ket{a} =
\frac{1}{\sqrt{2}}\left(\ket{ij} -\ket{ji}\right)$.

We calculate
\begin{widetext}
\begin{equation}
\begin{split}
\Tr_{\XX\otimes \YY}  
J\left(\1\otimes \rho^\top \right) &= 
\frac{1}{d^2-1}   \Tr_{\XX \otimes \YY}  
\left( 
\1 \otimes \left( \frac{1}{d^2}  \1 - \frac{1}{d} S \right)
+ T \otimes \left( S - \frac{1}{d} \1 \right)
 \right) 
\left(\1\otimes \rho \right)  \\
&= \frac{1}{d^2-1} \left( 
\1  \cdot\frac{1}{d} \Tr \left(  1 + \frac{1}{d}  \right) \rho
- T  \cdot \Tr \left( 1 + \frac{1}{d}  \right) \rho
 \right)   \\
&= \frac{1}{d^2-1} \left( 
\frac{1}{d} \left(  1 + \frac{1}{d}  \right) \1
-  \left( 1 + \frac{1}{d}  \right) T
 \right)   \\
&= \frac{1}{d(d-1)} \left( \1 - T\right)  - 
\frac{1}{d^2}  \1 .
\end{split}
\end{equation}
\end{widetext}
Simple calculations show that the trace norm of the above matrix equals $\frac{2}{d}$,
hence the upper bound is saturated. The construction of the input stat $\rho$ also
shows that this upper bound can be achieved without an additional register.

%
\end{proof}

\subsection{Asymmetric discrimination}
In the asymmetric discrimination we will consider two types of errors
separately. We would like to verify whether measurements in both black boxes are
the same (which corresponds to $H_0$ hypothesis) or they are different (which
corresponds to $H_1$ hypothesis). Formally, when the measurement in the first
black box, $\PP_0$, is unknown, we say that $\PP_0 = \int \PP_U dU$. The
measurement in the second black box can be either the same as in the first black
box ($\PP_? = \PP_0$) or it can be some other measurement, that is  $\PP_? =
\int \PP_V dV$. When performing asymmetric discrimination, we prepare an input
state $\ket{\psi} \in \XX \otimes \YY \otimes \ZZ$. If in both black boxes there
were the same measurements, then the output state yields $\rhozeropsi = \int
\left(\PP_U \otimes \PP_U \otimes \1_\ZZ \right) (\psi) dU. $ If the
measurements in the black boxes were different, when the output state is $
\rhoonepsi = \int \left(\PP_U \otimes \PP_V \otimes \1_\ZZ \right) (\psi) dU dV.
$ Next, we measure the output state by a binary measurement $\{\Omega, \1 -
\Omega \}$. We will focus on the case when the type I error cannot occur. The
optimal probability of the type II error is formulated as the following theorem.

\begin{theorem}\label{thm:both_random_asymmetric}
Let $\PP_0$ be a reference von Neumann measurement of dimension $d$ given 
without classical description. Let $\PP_1$ be another von Neumann measurement 
of the same dimension, also given without classical description. Consider the 
hypotheses testing problem described in Subsection~\ref{sec:asymmetric_scheme}.
Let $H_0$ hypothesis state that $\PP_? = \PP_0$ and let the alternative $H_1$ 
hypothesis state that $\PP_? = \PP_1$.
If no false positive error can occur, then the optimal probability of false 
negative error yields
\begin{equation}
\ptwo = 1 - \frac{1}{d}.
\end{equation}
Moreover, no additional register is needed to obtain this value.
\end{theorem}

\begin{proof}
As the input state to the discrimination procedure we take some state 
$\ket{\psi} \in \XX \otimes \YY$. Note that we assumed that this state is only 
on two registers. 
In this proof we will calculate the probability of the type II error assuming 
that the register $\ZZ$ is trivial. Later, we will prove that this gives the 
optimal probability and the additional register is not needed. 

If both measurements are the same, then the output state will be
\begin{equation}
\rhozeropsi = \int \left(\PP_U \otimes \PP_U\right) (\psi) dU. 
\end{equation}
If the measurements in the black boxes are different, then the output 
state will be
\begin{equation}
\rhoonepsi = \int \left(\PP_U \otimes \PP_V\right) (\psi) dU dV.
\end{equation}

We begin with calculating $\int \left(\PP_U \otimes \PP_U\right) (\psi) dU$ 
by the use of formula for recovering the action of a quantum channel given its 
Choi matrix.
Using the formula for the Choi matrix from Eq.~\eqref{eq:J_0}  and using the 
notation $T \coloneqq \Delta(S)$ we calculate

\begin{widetext}
\begin{equation}
\begin{split}
\rhozeropsi&= \Tr_{\ZZ} \left( J \left(  
\int \PP_U \otimes \PP_U dU 
\right) \left(\1 \otimes \psi^\top \right)  \right) \\
&= \frac{1}{d (d^2-1)} \left(
\left(d - \tr \left(S \psi^\top\right) \right) \1 
+  \left( d \tr \left(S \psi^\top\right)  - 1 \right) T
\right). \\
\end{split}
\end{equation}
\end{widetext}

Let us take the input state to be antisymmetric, that is it satisfies 
$\tr \left( S \psi^\top\right)=-1$. We calculate
\begin{equation}
\begin{split}
\rhozeropsi
&= \frac{1}{d (d^2-1)} 
\left( \left( d + 1 \right) \1 -  \left( d + 1 \right) T \right) \\
& = \frac{1}{d (d-1)} \left(  \1 -  T \right). 
\end{split}
\end{equation}
By similar calculation, using the antisymmetric input state we have
\begin{equation}
\begin{split}
\rhoonepsi& = \Tr_{\ZZ} \left( J \left(  
\int \PP_U \otimes \PP_V dU 
\right) \left(\1 \otimes \psi^\top \right)  \right) \\
& = \Tr_{\ZZ} \left(  \left(\frac{1}{d^2} \1 \otimes \1\right)
\left(\1 \otimes \psi^\top \right)  \right) \\
&= \frac{1}{d^2} \Tr_{\ZZ} \left(\1 \otimes \psi^\top   \right) 
=\frac{1}{d^2} \1.
\end{split}
\end{equation}
As the measurement effect we take $\Omega \coloneqq \1 - T$.
Hence
\begin{equation}
\begin{split}
\pcondone &= 1- \tr\left( \Omega \rhozeropsi \right)\\
& =1- \frac{1}{d (d-1)} 
\tr \left(\left(  \1 - T \right) \left(  \1 - T \right) \right)\\
& = 0, 
\end{split}
\end{equation}
and
\begin{equation}
\begin{split}
\pcondtwo &= \tr\left( \Omega \rhoonepsi \right) = \frac{1}{d^2} \tr 
\left(\1 - T \right) \\
& = \frac{d(d-1)}{d^2} = 1-\frac{1}{d}.
\end{split}
\end{equation}

From Appendix~\ref{app:inequality_between_errors} we know that the probability
of erroneous discrimination is the symmetric scheme (which equals
$1-p^H_{succ}$) is never bigger than the arithmetic mean of probabilities of the
type I and type II errors. As
\begin{equation}
\frac{1}{2} \left(\pcondone + \pcondtwo\right) = \frac{1}{2} - \frac{1}{2d},
\end{equation}
then we conclude that our value of $\pcondtwo = 1-\frac{1}{d}$ is optimal and 
hence $\ptwo = \pcondtwo$.

Finally, note  that the optimal value $\ptwo$ can be achieved for the input state 
$\ket{\psi} \in \XX \otimes \YY$, that is when the register $\ZZ$ is trivial. 
Hence, the additional register is not needed for asymmetric discrimination in 
this case.
\end{proof}

\section{Discrimination between a fixed and unknown von Neumann 
measurements}\label{sec:one_fixed_measurement}

In this section we assume that instead of the unknown reference measurement 
from the previous section, we are given $\PP_0$ as a fixed von Neumann 
measurement $\PP_U$. 
We will begin with studying symmetric discrimination and later proceed to studying the asymmetric discrimination scheme.

\subsection{Symmetric discrimination}
Now we focus on the situation when we want to distinguish between a fixed von 
Neumann measurement $\PP_U$ and a Haar-random measurement $\int \PP_V dV$. The probability of successful discrimination is formulated as a theorem. 

\begin{theorem}\label{thm:one_random_measurement_symmetric}
Let $\PP_0 = \PP_U$ be a reference von Neumann measurement of dimension $d$. 
Let $\PP_1$ be another von Neumann measurement 
of the same dimension, but given without classical description. The optimal 
probability of correct verification whether $\PP_1 = \PP_0$ or $\PP_1 \neq 
\PP_0$ in the scheme described in Subsection~\ref{sec:symmetric_scheme} 
equals  
\begin{equation}
p^H_{succ}  = 1- \frac{1}{2d}.
\end{equation}
\end{theorem}

\begin{proof}
Without loss of generality we can take $U = \1$.  
To calculate the bound from Holevo-Helstrom theorem~\eqref{eq:HH}, we want to 
calculate the diamond norm distance between quantum measurements
\begin{equation}
\left\| \PP_\1 \otimes \PP_\1 - \PP_\1 \otimes \int \PP_V dV \right\|_\diamond.
\end{equation}
Using properties of the diamond norm~\cite{watrous2018theory} we calculate 
\begin{widetext}
\begin{equation}
\begin{split}
\left\| \PP_\1 \otimes \PP_\1 - \PP_\1 \otimes \int \PP_V dV \right\|_\diamond
&= \left\| \PP_\1 \otimes \left(\PP_\1 - \int \PP_V dV\right) \right\|_\diamond 
\\
&= \left\| \PP_\1 \right\|_\diamond  
\left\|  \PP_\1 - \int \PP_V dV \right\|_\diamond \\
&= \left\|  \PP_\1 - \int \PP_V dV \right\|_\diamond. 
\end{split}
\end{equation}
\end{widetext}
To do this, we use the fact that $\PP_V = \Delta \Phi_{V^\dagger}$.
Moreover, we know that $J(\Phi_\1) = \projV{\1}$ and $J(\Phi_\star) = \1/d$,  
where $\Phi_\star$ is the depolarizing channel defined in 
Eq.~\eqref{eq:depolarizing_channel_df}.
Therefore, calculating directly both lower and upper bounds for the diamond norm from 
Eq.~\eqref{eq:diamond_bounds}, 
we obtain
\begin{equation}
\left\| \PP_\1 -  \int \PP_V dV \right\|_\diamond
= 2-\frac{2}{d}.
\end{equation}
Finally
\begin{equation}
p^H_{succ} = \frac{1}{2} + \frac{1}{4} 
\left( 2-\frac{2}{d} \right)= 1- \frac{1}{2d}.
\end{equation}
\end{proof}

\subsection{Asymmetric discrimination}
In this subsection we will focus on asymmetric discrimination between a fixed 
von Neumann measurement $\PP_U$ and a Haar-random measurement $\PP_V$. 
We will be interested in the scenario when the false positive error cannot occur. The optimized probability of the false negative error is formulated as a theorem.

\begin{theorem}
Let $\PP_0 = \PP_U$ be a fixed von Neumann measurement and $\PP_1$ be some other
von Neumann measurement given without classical description. Let the $H_0$
hypothesis correspond to the case when $\PP_? = \PP_0$ and $H_1$ hypothesis
correspond to the case when $\PP_? = \PP_1$. Consider the discrimination 
scheme described in Subsection~\ref{sec:asymmetric_scheme}. If no false 
positive error can occur, then the optimal probability of false negative error 
yields
\begin{equation}
\ptwo = \frac{1}{d}.
\end{equation}
\end{theorem}

\begin{proof}
This proof goes similar as the proof of
Theorem~\ref{thm:both_random_asymmetric}. We will choose a fixed
input state  on only two registers. We will also fix the final measurement and
calculate the probabilities of making the false positive and false negative
errors. Later, from inequality between errors in symmetric and asymmetric
schemes in Appendix~\ref{app:inequality_between_errors} we will see that the
calculated $\ptwo$ is the optimal one.

As the input state we take $\psi \coloneqq \frac{1}{d} \projV{\1}$.
We calculate the output states
\begin{equation}
\begin{split}
\rhozeropsi &\coloneqq 
\left(\PP_U \otimes \1\right) (\psi)
= \frac{1}{d} \left(\PP_U \otimes \1\right) (\projV{\1})\\
&\phantom{:}= \frac{1}{d} \sum_{i} \proj{i} \otimes \proj{u_i}^\top  
\end{split}
\end{equation}
and
\begin{equation}
\begin{split}
\rhoonepsi &\coloneqq 
\int \left(\PP_V \otimes \1\right) (\psi) dV\\
&= \frac{1}{d} \int \left( \PP_V \otimes \1\right) (\projV{\1}) dV \\
&= \frac{1}{d} \int \sum_{i} \proj{i} \otimes \proj{v_i}^\top dV \\
&=  \frac{1}{d}  \sum_{i} \proj{i} \otimes \int \proj{v_i}^\top dV\\
&=  \frac{1}{d^2}   \1 \otimes \1.
\end{split}
\end{equation}

Recall that the  measurement effect $\Omega$ correspond to $H_0$ hypothesis and
$\1 - \Omega$  correspond to $H_1$ hypothesis. Hence we have probabilities of
false positive and false negative errors (for given input state) equal
\begin{equation}
\begin{split}
\pcondone &= 1- \tr \left(\Omega \rhozeropsi \right),\\
\pcondtwo &=  \tr \left(\Omega \rhoonepsi \right).
\end{split}
\end{equation}

Without loss of generality we can consider $\Omega$ in the block-diagonal form, 
ie. 
\begin{equation}
\Omega \coloneqq \sum_{i} \proj{i} \otimes \Omega_{i}^\top.
\end{equation}
As the unitary matrix $U$ is known, we can use it to construct the final 
measurement. Let 
\begin{equation}
\Omega_{i} \coloneqq \proj{u_i}
\end{equation}
for every $i = 1, \ldots , d$.

Then
\begin{widetext}
\begin{equation}
\begin{split}
\tr \left(\Omega \rhozeropsi  \right)
&= \tr \left(\left(\sum_{i} \proj{i} \otimes 
\proj{u_i}^\top \right)
\left(\frac{1}{d} \sum_{j} \proj{j} \otimes \proj{u_j }^\top \right)\right) \\
&= \frac{1}{d} \sum_{i} \tr \left(\ket{u_i} \braket{u_i}{u_i} \bra{u_i}\right) 
= \frac{1}{d} \sum_{i} |\braket{u_i}{u_i}|^2 = 1  
\end{split}
\end{equation}
\end{widetext}
and hence
\begin{equation}
\pcondone = 1- \tr \left(\Omega \rhozeropsi \right) = 0.
\end{equation}

Eventually
\begin{equation}
\begin{split}
\pcondtwo &=  \tr \left(\Omega \rhoonepsi \right)\\
&= \tr \left(\left(\sum_{i} \proj{i} \otimes \proj{u_i}^\top \right)
\left(\frac{1}{d^2} \1 \otimes \1 \right) \right) \\
&= \frac{1}{d^2} \sum_i \tr\left( \proj{u_i}\right) 
= \frac{1}{d}.
\end{split}
\end{equation}

It remains to explain why $\pcondtwo = \ptwo$. Note that the arithmetic mean of 
probabilities of both types of errors equals $\frac{1}{2d}$ which is equal to 
the probability of erroneous discrimination in the symmetric scheme (see 
Theorem~\ref{thm:one_random_measurement_symmetric}).  From the inequality 
between errors in the symmetric and asymmetric schemes in 
Appendix~\ref{app:inequality_between_errors} we conclude that $\ptwo = 
\frac{1}{d}$.

\end{proof}

\section{Conclusion}\label{sec:conclusion}
We were studying the problem whether the given von Neumann measurement is the
same as the reference one. We were considering the situation when the reference
measurement is given without classical description and when its classical
description is known. Both situations were studied in the symmetric and
asymmetric scenarios. We proved that in both cases one can achieve the
probability of false positive error equal zero, and we calculated optimal
probabilities of false negative errors. We also calculated the probabilities of
successful discrimination in the symmetric discrimination scheme.

\section*{Acknowledgements}
ZP and AK were supported by the project ,,Near-term quantum computers
Challenges, optimal implementations and applications” under Grant Number
POIR.04.04.00-00-17C1/18-00, which is carried out within the Team-Net programme
of the Foundation for Polish Science co-financed by the European Union under the
European Regional Development Fund. LP acknowledges support from the National
Science Center (NCN), Poland, under Project Sonata Bis 10, No.
2020/38/E/ST3/00269.

\bibliographystyle{quantum}
\bibliography{bibliography}

\appendix

\section{Lemma~\ref{lm:j2=wj2_for_measurements}}\label{app:lemma}

\begin{lemma}\label{lm:j2=wj2_for_measurements}
Let $J$ be as defined in 
Eq.~\eqref{eq:def_J_for_both_unknown_measurements}, $T \coloneqq \Delta (S)$  
and $W \coloneqq (2T - \1) \otimes S$, where $S$ is the swap matrix of 
dimension $d^2$. Then $J^2 = (WJ)^2$.
\end{lemma}

\begin{proof}
As
\begin{widetext}
\begin{equation}
J^2 = \left(\frac{1}{d^2-1}\right)^2
\left(\frac{1}{d^2} \1 \otimes \1 - \frac{1}{d} \1 \otimes S + T \otimes S - 
\frac{1}{d} T \otimes \1 \right)^2,
\end{equation}
\end{widetext}
we calculate
\begin{widetext}
\begin{equation}
\begin{split}
&\left(\frac{1}{d^2} \1 \otimes \1 - \frac{1}{d} \1 \otimes S + T \otimes S - 
\frac{1}{d} T \otimes \1 \right)^2 \\
&= \frac{1}{d^4} \1 \otimes \1 - \frac{1}{d^3} \1 \otimes S + \frac{1}{d^2} T 
\otimes S - \frac{1}{d^3} T \otimes \1  \\
&-\frac{1}{d^3} \1 \otimes S + \frac{1}{d^2} \1 \otimes \1 - \frac{1}{d} T 
\otimes \1 + \frac{1}{d^2} T \otimes S  \\
&+ \frac{1}{d^2} T \otimes S - \frac{1}{d} T \otimes \1 +  T 
\otimes \1 - \frac{1}{d} T \otimes S \\
&-\frac{1}{d^3} T \otimes \1 + \frac{1}{d^2} T \otimes S - \frac{1}{d} T 
\otimes S + \frac{1}{d^2} T \otimes \1  \\
&= \frac{d^2+1}{d^4} \1 \otimes \1 - \frac{2}{d^3} \1 \otimes S 
+ \left( 1 + \frac{1}{d^2} - \frac{2}{d}-\frac{2}{d^2}  \right) T \otimes \1 
+ \left( \frac{4}{d^2} - \frac{2}{d} \right)  T \otimes S \\
&= \frac{d^2+1}{d^4} \1 \otimes \1 - \frac{2}{d^3} \1 \otimes S 
+ \frac{(d^2+1) (d-2)}{d^3}  T \otimes \1 
+ \frac{4-2d}{d^2} T \otimes S,
\end{split}
\end{equation}
\end{widetext}
and eventually
\begin{widetext}
\begin{equation}
J^2 = \left(\frac{1}{d^2-1}\right)^2
\left(\frac{d^2+1}{d^4} \1 \otimes \1 - \frac{2}{d^3} \1 \otimes S 
+ \frac{(d^2+1) (d-2)}{d^3}  T \otimes \1 
+ \frac{4-2d}{d^2} T \otimes S
\right).
\end{equation}
\end{widetext}

On the other hand
\begin{widetext}
\begin{equation}
WJ = \left( 2 T \otimes S - \1 \otimes S  \right)
\frac{1}{d^2-1}
\left(\frac{1}{d^2} \1 \otimes \1 - \frac{1}{d} \1 \otimes S + T \otimes S - 
\frac{1}{d} T \otimes \1 \right).
\end{equation}
\end{widetext}
Hence we calculate
\begin{widetext}
\begin{equation}
\begin{split}
& \left( 2 T \otimes S - \1 \otimes S  \right)
\left(\frac{1}{d^2} \1 \otimes \1 - \frac{1}{d} \1 \otimes S + T \otimes S - 
\frac{1}{d} T \otimes \1 \right) \\
&= \frac{2}{d^2} T \otimes S  - \frac{2}{d} T \otimes \1 + 2 T \otimes \1
- \frac{2}{d} T \otimes S \\
&- \frac{1}{d^2} \1 \otimes S + \frac{1}{d} \1 \otimes \1 -  T 
\otimes \1 + \frac{1}{d} T \otimes S  \\
&= \frac{1}{d} \1 \otimes \1  - \frac{1}{d^2} \1 \otimes S + \frac{d-2}{d} T 
\otimes \1 - \frac{d-2}{d^2} T \otimes S,
\end{split}
\end{equation}
\end{widetext}
and thus
\begin{widetext}
\begin{equation}
\begin{split}
&\left(  \frac{1}{d} \1 \otimes \1  - \frac{1}{d^2} \1 \otimes S + 
\frac{d-2}{d} T \otimes \1 - \frac{d-2}{d^2} T \otimes S \right)^2 \\
&= \frac{1}{d^2} \1 \otimes \1  - \frac{1}{d^3} \1 \otimes S + \frac{d-2}{d^2} 
T \otimes \1 - \frac{d-2}{d^3} T \otimes S \\
&- \frac{1}{d^3} \1 \otimes S + \frac{1}{d^4} \1 \otimes \1 - \frac{d-2}{d^3} T 
\otimes S +\frac{d-2}{d^4} T \otimes \1 \\
&+ \frac{d-2}{d^2} T \otimes \1 - \frac{d-2}{d^3} T \otimes S
+ \frac{(d-2)^2}{d^2} T \otimes \1 - \frac{(d-2)^2}{d^3} T \otimes S \\
&- \frac{d-2}{d^3} T \otimes S + \frac{d-2}{d^4} T \otimes \1
- \frac{(d-2)^2}{d^3} T \otimes S + \frac{(d-2)^2}{d^4} T \otimes \1 \\
&= \frac{d^2+1}{d^4} \1 \otimes \1 - \frac{2}{d^3} \1 \otimes S + 
\frac{(d^2+1)(d-2)}{d^3} T \otimes \1 + \frac{4 - 2d}{d^2} T \otimes S.
\end{split}
\end{equation}
\end{widetext}
Eventually
\begin{widetext}
\begin{equation}
(WJ)^2 = \left(\frac{1}{d^2-1}\right)^2
\left(\frac{d^2+1}{d^4} \1 \otimes \1 - \frac{2}{d^3} \1 \otimes S 
+ \frac{(d^2+1) (d-2)}{d^3}  T \otimes \1 
+ \frac{4-2d}{d^2} T \otimes S
\right)
\end{equation}
\end{widetext}
and hence $(WJ)^2 = J^2$.
\end{proof}

\section{Inequality between errors}\label{app:inequality_between_errors}

We will show that 
\begin{equation}
p_e^H \leq  \frac{1}{2} (p_1 + p_2),
\end{equation}
where $p_e^H = 1-p_{succ}^H$ is the probability of error from the 
Holevo-Helstrom Theorem.

Let us recall that from Holevo-Helstrom Theorem we have
\begin{equation}
\frac{1}{2} \Tr(\Omega_0\rho_0) + \frac{1}{2} \Tr (\Omega_1\rho_1) 
\leq 1- p_e^H,
\end{equation}
hence
\begin{equation}
p_e^H \leq 1-\frac{1}{2} \left(\Tr(\Omega_0\rho_0) +\Tr(\Omega_1\rho_1)\right).
\end{equation}

On the other hand we know that
\begin{equation}
\begin{split}
\Tr(\Omega_0\rho_0) + \Tr(\Omega_1\rho_0)
&= 1, \\
\Tr(\Omega_0\rho_1) + \Tr(\Omega_1\rho_1)
&= 1.
\end{split}
\end{equation}
and hence
\begin{equation}
\Tr(\Omega_0\rho_0) +\Tr(\Omega_1\rho_1) 
= 2 - (p_1 + p_2).
\end{equation}

Therefore
\begin{equation}
\begin{split}
p_e^H &\leq 
1-\frac{1}{2} \left(\Tr(\Omega_0\rho_0) +\Tr(\Omega_1\rho_1)\right)\\
&= 1-\frac{1}{2} \left(  2-(p_1 + p_2)   \right) \\
&= \frac{1}{2} (p_1 + p_2).
\end{split}
\end{equation}

\end{document}